\documentclass{article}
\usepackage{fullpage}

\usepackage{amsmath}
\usepackage{amssymb}
\usepackage{amsthm}
\usepackage{enumitem}
\usepackage{subcaption}
\usepackage{color}
\usepackage{graphicx}

\newcommand{\ruyquote}[1]{\refstepcounter{equation} \vspace{0.15cm}%
     \parbox{13cm}{\em #1}\hspace*{0.5cm}($\theequation$)\\[0.15cm]}

\newtheorem{theorem}{Theorem}
\newtheorem{lemma}[theorem]{Lemma}
\newtheorem{corollary}[theorem]{Corollary}
\newtheorem{observation}{Observation}

\title{Filming runners with drones is hard}
\author{José-Miguel Díaz-Báñez\thanks{Department of Applied Mathematics, University of Seville, SPAIN. \texttt{dbanez@us.es}} \and Ruy Fabila-Monroy\thanks{Departamento de Matem\'aticas, Cinvestav, MEXICO. \texttt{ruyfabila@math.cinvestav.edu.mx}} }

\begin{document}

\maketitle

\begin{abstract}
The use of drones or Unmanned Aerial Vehicles (UAVs) for aerial photography and cinematography is becoming widespread.
The following optimization problem has been recently considered. Let us imagine a sporting event where a group of runners are competing and a team of drones with cameras are used to cover the event. The media \emph{director} selects a set of \emph{filming scenes} (determined by locations and time intervals) and the goal is to maximize the total \emph{filming time} (the sum of recordings) achieved by the aerial cinematographers.  Recently, it has been showed that this problem can be solved in polynomial time assuming the drones have unlimited battery endurance. In this paper, we prove that the problem is NP-hard for the more realistic case in which the battery endurance of the drones is limited. 
\end{abstract}



\section{Introduction}
\label{intro}

The rapid proliferation of commercial unmanned aerial vehicles (UAVs), commonly referred to as \emph{drones}, has had a strong impact on media production and coverage.  UAVs, which partially replace motorcycles and helicopters in the field, provide an affordable and flexible alternative for quickly capturing impressive aerial footage in a variety of scenarios, including movie and TV shoots, as well as live outdoor event coverage. Therefore, optimization problems aimed at enhancing the efficiency of UAV-based shooting represent an emerging research area with substantial potential for the industry.

There are examples of algorithms in the robotics literature for maximum coverage of targets
from appropriate viewpoints
using multiple cooperating UAVs.
For example, in \cite{nageli2017real}, an online
planning algorithm that jointly optimizes feasible trajectories and control
inputs for multiple UAVs is proposed.
Also,  a method 
to ensure that cinematographic properties and dynamic constraints are ensured along the trajectories is proposed in \cite{galvane2018directing}. Non-linear optimization is applied to generate polynomial curves with minimum
curvature variation, accounting for target visibility and collision
avoidance. The motion of multiple UAVs around dynamic targets
is coordinated by means of a centralized master–slave approach
to solve conflicts.
A novel formulation for non-linear trajectory
optimization is proposed in \cite{alcantara2021optimal}.
The method integrates the aesthetic
quality of the videos as well as the constraints imposed
by drone dynamics.

However, there has been limited effort dedicated to investigating optimal trajectory planning problems involving multiple cooperating UAVs
under
battery autonomy limitations.
In this framework, the following scheduling problem has been conjectured to be NP-hard in \cite{caraballo2020autonomous}:  Consider an outdoor event to be filmed by a set of $k$ drones equipped with cameras and limited battery endurance.
The media director specifies a set of shooting actions (scenes) determined by waypoints and time intervals during which the drones should film the moving targets (e.g., cyclists, runners, etc.).  
A scene, or part of it, can be filmed by one or multiple drones, meaning a drone may film the first part of a scene, while the rest may be filmed by a different one. 
The objective is to compute a schedule for the team that maximizes the total filming time.
There is one base station from which the drones start and return to recharge their batteries at any time. Battery recharging is assumed to be instantaneous. It is also assumed that a collision-free path can be computed between any pair of locations, and given this path, the required travel time and battery consumption can be estimated. Consequently, it can be determined at any moment whether a drone has enough battery to return to the base.
The authors in  \cite{caraballo2020autonomous} proposed a discretized version of the problem, which enables the approximation of the above optimization problem in polynomial time. They assumed a discrete set of possible entry and exit points within the time intervals. Their approach is based on constructing a directed acyclic graph, with vertices representing pairs $(p, t)$, where $p$ denotes a position and $t$ corresponds to an associated instant of time. Discretization was performed based on time rather than distance to align with the application's goal of maximizing filming time, thereby focusing on the total number of discrete time segments covered. The authors introduced a graph-based algorithm relying on dynamic programming to find an optimal solution for the discrete problem concerning a single drone in polynomial time. A greedy strategy was subsequently employed to address the problem for multiple drones sequentially.



Several variants of this challenging problem
in cinematography have been recently considered.
A non-linear formulation recently proposed in \cite{alcantara2021optimal}.  Unlike \cite{caraballo2020autonomous}, they assume that the director has designed a mission comprising several shots, with a specific shot pre-assigned to each UAV.  Their objective is to plan trajectories in order to execute
all shots online in a coordinated manner by a small team of drones.

On the other hand, in  \cite{aichholzer2020scheduling} it has been proved that the mentioned scheduling problem can be solved in polynomial time when the battery endurance is considered unlimited, that is, the drones do not require to recharge their batteries. The authors were able to translate the problem to the problem of finding a maximum weight matching in a weighted bipartite graph.
They also showed how to efficiently solve the problem for one drone with limited battery endurance.
In this work, we prove the NP-hardness of the problem for a team of drones with limited battery endurance.

\section{Problem statement}
\label{problem}
Consider a scenario in which an outdoor event is going to be filmed by a set of drones. A cinematography director specifies certain locations and time intervals at
which the filming of certain scenes is desired. This is represented as 
 a set of $n$ tuples \[\mathcal{F}:=\left \{(p_i,I_i):1 \le i \le n \right \},\] where $p_i$ is a point in the plane
and $I_i$ is a time interval. $\mathcal{F}$ is called the \emph{film plan}, and each $(p_i,I_i)$  is called 
a \emph{scene}. A scene $(p_i,I_i)$, or a portion of it, can be filmed by a
drone located at $p_i$ during (part of) the time interval $I_i$. 
Each drone flies at unit speed  and has a limited \emph{battery endurance} of $L$. 
 All drones start at a point  $p_0$ in the plane called the \emph{base}. 
Upon returning to the base, we assume that the drone swaps instantaneously its battery for a fully charged one; 
while at the base, the drone does not consume its battery.


A \emph{flight path} for drone $i$ is a sequence of tuples $(p_{i(0)}, J_{i(0)}), \dots, (p_{i(m)}, J_{i(m)})$ with:
\begin{itemize}
\item $0 \le i(j) \le n$ for all $0 \le j \le m$;
\item $i(0) = i(m) = 0$; and
\item the $J_{i(j)}$ are time intervals with $J_{i(j)} \subset I_{i(j)}$, if $i(j) \neq 0$.
\end{itemize}
A flight path represents a schedule that a drone might follow. A flight path is \emph{realizable} if:
\begin{itemize}
\item for every $0 \le j \le m-1$, the drone can travel from $p_{i(j)}$ at the end of interval $J_{i(j)}$ to $p_{i(j+1)}$ at the beginning of interval $J_{i(j+1)}$; and
\item the drone is never away from the base for a time interval longer than $L$.
\end{itemize}

 A \emph{flight plan} $F$ is a set of realizable flight paths. Its \emph{filming time} is defined as 
 \[\sum_{i=1}^{n} \left |\bigcup_{\substack{P \in F \\ (p_i,J_i) \in P }} \left(I_i \cap J_i \right )\right |.\]
Thus, the filming time of $F$ is the total time of the film plan that can be captured by $|F|$ drones, each implementing
one of the flight paths in $F$. 

In this paper we consider the following decision problem.
\bigbreak
\noindent \textbf{Filming Time with Limited  Battery (FTLB) problem}
\begin{quote} \emph{Given $(\mathcal{F},k,L,T)$, does there exists a flight plan for $\mathcal{F}$, with $k$ drones, each with
battery endurance equal to $L$, of filming time at least $T$?}\end{quote}
Our main result is the following.
\begin{theorem} \label{thm:main}
The FTLB problem is NP-Complete.
\end{theorem}

\section{NP-hardness}

The FTLB problem is clearly in NP.
In this section we prove that the  FTLB problem is NP-Complete by a reduction from 3-SAT.
In~\cite{itai}, the authors show that the following problem is NP-Complete. 
\begin{quote} 
\emph{Given a graph $G$ with two distinct vertices $s$ and $t$, 
find the maximum number of interior vertex disjoint paths of length at most five between $s$ and $t$.} 
\end{quote}
Our reduction
was inspired by the reduction of~\cite{itai}. However, the geometric nature of our setting makes our construction much more elaborated. 

Given  an instance $\varphi$ of 3-SAT, with variables $x_1,\dots,x_n$,
 an instance $\varphi'$ of 3-SAT, on the same variables, can be constructed in polynomial time,
such that the number of times that $x_i$ appears on $\varphi'$ is equal to the number of times
$\overline{x_{i}}$ appears on $\varphi'$(see~\cite{itai}). Moreover, $\varphi$ is satisfiable if and only if $\varphi'$
is satisfiable. 
In what follows, let $\varphi$ be an instance of 3-SAT with variables $x_1,\dots,x_n$, clauses $C_1,\dots,C_l$, and in which $x_i$ and $\overline{x_i}$ each appear  exactly
$m_i$ times; let \[m:=\sum_{i=1}^n m_i. \] 
If necessary, we duplicate every clause at most two times so that $m_i \ge 3$ for all $i$.
We construct in polynomial time an instance $(\mathcal{F},k,L,T)$ of the FTLB problem such that:

\begin{lemma}\label{th:satisf}
$\varphi$ is satisfiable
if and only if there is a flight plan for $\mathcal{F}$,  with $k$ drones, each with battery endurance equal to $L$, of filming time at least $T$.
\end{lemma}

\noindent Theorem~\ref{thm:main} follows directly from Lemma~\ref{th:satisf}. The rest of the paper
is devoted to prove Lemma~\ref{th:satisf} 

\subsection{Construction}
Throughout our construction we define various parameters; afterwards, we give lower or upper bounds on these parameters so that certain conditions are met.
These conditions enable us to prove Lemma~\ref{th:satisf}

  \subsection*{ Variable Gadgets}
  
  For every variable $x_i$ we construct a gadget as follows. Let $S_1$ and $S_2$ be two
 concentric circles of radii $r_1$ and $r_2$, respectively, with $r_1 > r_2$.
 Let $x_{i1}, \overline{x}_{i1},x_{i2}, \overline{x}_{i2},\dots,x_{im_i}, \overline{x}_{im_i}$ be the vertices,
 in clockwise order, of a  regular polygon of $2m_i$ sides on $S_2$.
  The vertices $x_{ij}$ and  $\overline{x}_{ij}$ correspond to the $j$-th time that literals $x_i$ and $\overline{x_i}$ appear on $\varphi$, respectively.
  We call these vertices, \emph{literal vertices}.
 For every $j=1,\dots,m_i$, place on $S_1$ a vertex $y_{ij}$ equidistant to $x_{ij}$ and $\overline{x}_{ij}$.
 We call these vertices, \emph{$Y$-vertices}.
 For every $j=1,\dots,m_i$, place on $S_1$ a vertex $w_{ij}$ equidistant to $\overline{x_{ij}}$ and $x_{i(j+1)\mod m_i}$. 
 We call these vertices, \emph{$W$-vertices}.
 See Figure~\ref{fig:cycling_gadgets}, zoom. Suppose for the time being that $r_2$ is such that the $x_{ij}$ and $\overline{x_{ij}}$
 vertices
 lie at the midpoint of the edges of the regular polygon with vertices the $y_{ij}$ vertices. Since we are assuming that this
 polygon has at least $2m_i\ge 6$  vertices, the distance from $x_{ij}$ to $ \overline{x_{ij}}$ is at least
 $2\cos(\pi/6))=\sqrt{3}$ the distance from $x_{ij}$ (and $\overline{x_{ij}}$) to $y_{ij}$.
 Thus, we can choose $r_2$ so that
 \begin{equation}
  d(x_{ij},\overline{x_{ij}})=\frac{3}{2}d(x_{ij},y_{ij})=\frac{3}{2}d(\overline{x_{ij}},y_{ij}), \label{q:delta}
 \end{equation}
 These gadgets are called \emph{variable gadgets}.
 
 \begin{figure}
 	\centering
 	\includegraphics[width=.7\textwidth]{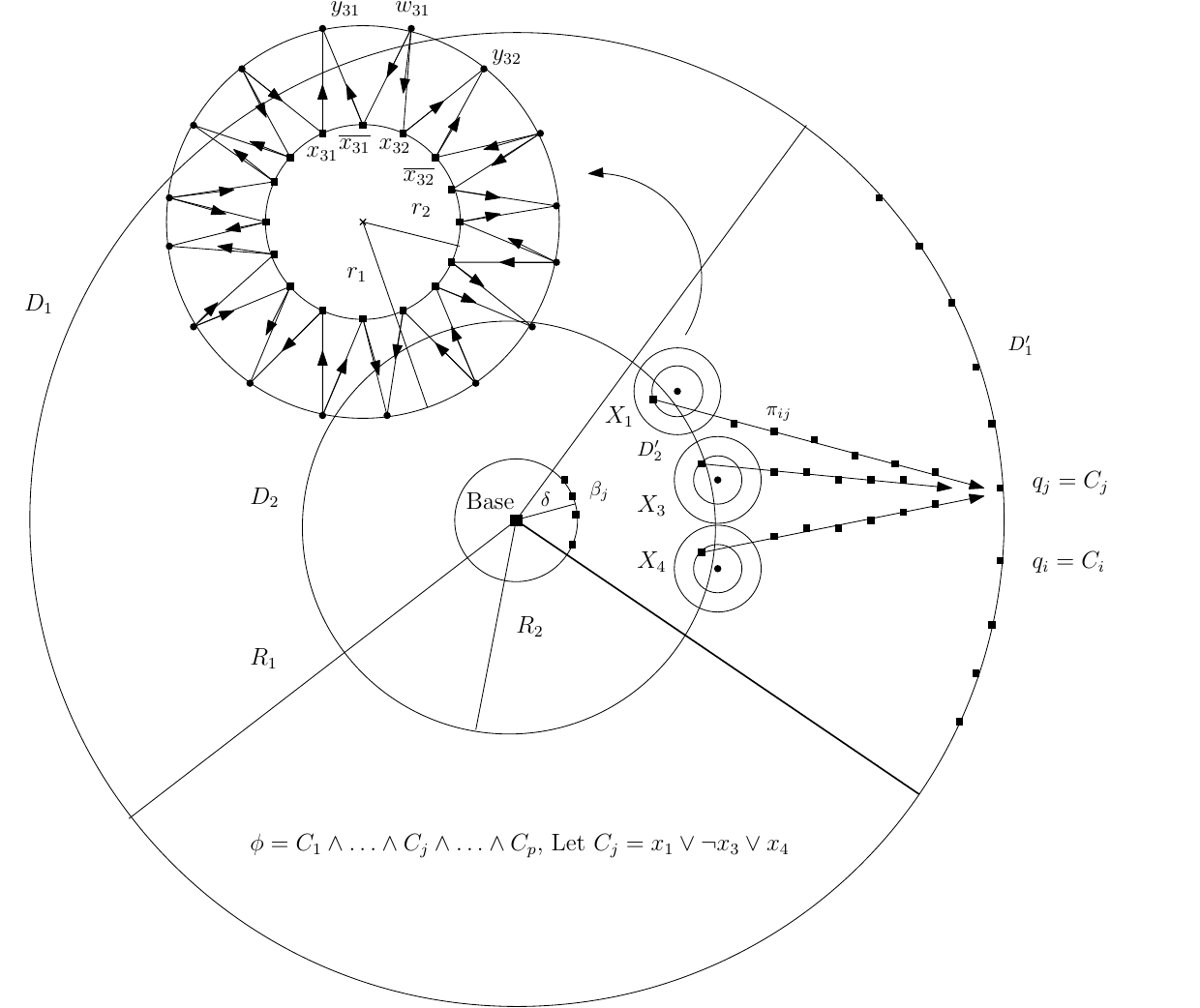}
 	\caption{Construction. The clause vertices $q_j$ are located at the exterior circle $D_1$ inside the arc $D_1'$. The gadgets are located at the inner circle $D_2$ inside the arc $D_2'$. The three directed edges correspond to the literals appearing in the clause $C_j$. Zoom: Variable Gadget. Auxiliaries vertices $y_{ij}$ and $w_{ij}$ are located in an alternating way in the gadget. }\label{fig:cycling_gadgets}
 \end{figure}
 
  \begin{figure}
 	\centering
 	\includegraphics[width=.7\textwidth]{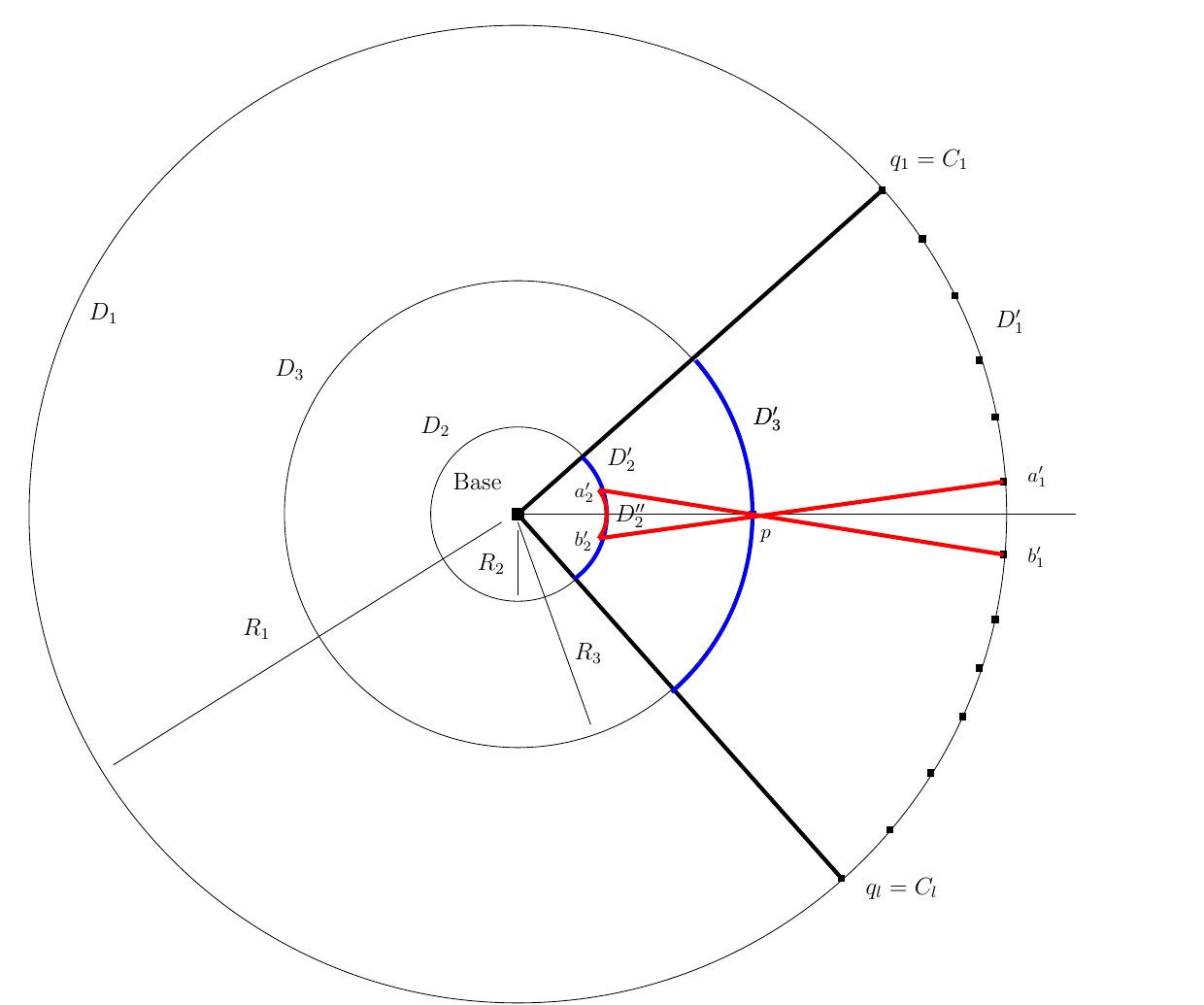}
 	\caption{The collision zone is the  blue region of the annulus. The variable vertices are located at $D_2''.$ The clause vertices are located at $D_1'$. Two intersecting segments connecting variable and clause vertices intersect inside the collision zone. }\label{fig:collision_zone}
 \end{figure}

 \subsection*{ Construction. Scene locations} 

Giving an instance of the 3-SAT, in this section we design the film plan. Unfortunately, in our problem we do not have edges between locations. To go around this issue, we carefully
place filming intervals (scenes) to force the drones to follow a set of fictitious paths. The construction is illustrated in Figures \ref{fig:cycling_gadgets} and  \ref{fig:collision_zone}.

We assume that the base is located at the origin of the coordinate system. Let $D_1$ and $D_2$ be two concentric circles centered at the origin. Let $R_1 > R_2$ be the radii of
$D_1$ and $D_2$, respectively. Let $a_1:=(R_1 \cos(-\pi/4), R_1 \sin(-\pi/4))$ and 
$b_1:=(R_1 \cos(\pi/4), R_1 \sin(\pi/4))$ be points on $D_1$. Let $D_1'$ be the circular arc in $D_1$ from
$a_1$ to $b_1$. Starting at $a_1$ and ending at $b_1$, place $q_1,\dots,q_l$, equally spaced vertices along $D_1'$. 
The distance between two such consecutive points is equal to \[\frac{\pi R_1}{2(l-1)}.\]
We call these vertices \emph{clause vertices}. For  $ 1\leq i\leq l$, vertex $q_i$ corresponds to the clause $C_i$. 

Let $a_2:=(R_2 \cos(-\pi/4), R_2 \sin(-\pi/4))$ and 
$b_2:=(R_2 \cos(\pi/4), R_2 \sin(\pi/4))$ be points on $D_2$.
Let $D_2'$ be the circular arc  in $D_2$ from
$a_2$ to $b_2$. Choose $R_1$ sufficiently large with respect to $R_2$, so that 

\ruyquote{the interior of any
line segment joining a point of $D_2'$ with a point of $D_1'$ does not intersect 
$D_2$. \label{q:no_int}}
Let $s$ be a line segment joining a point of $D_2'$ with a point of $D_1'$. Note that 
the length of $s$ is lower bounded by 
\begin{equation} \label{q:lower bound d}
    R_1-R_2.
\end{equation}
We choose $R_1$ sufficiently large with respect
to $R_2$ such that 

\ruyquote{ the length of $s$ is less than $R_1+r_1$. \label{q:upper bound d}}
\noindent Thus,
the absolute value of the difference in lengths of two such segments is less than $R_2+r_1$.
Choose $R_1$ sufficiently large with respect to $R_2$, so that 
\begin{equation}
 \frac{\pi R_1}{2(l-1)}\ge 3 R_2. \label{q:clauses_separated}
\end{equation}

Let $a_1':=(R_1 \cos(-\pi/(4(l-1)), R_1 \sin(-\pi/4(l-1)))$ and 
$b_1':=(R_1 \cos(\pi/(4(l-1))), R_1 \sin(\pi/(4(l-1)))$ be two points on $D_1$.
Let $D_1''$ be the circular arc in $D_1$ from
$a_1'$ to $b_1'$.
Let $R_3 > R_2$, such that
\begin{equation}\label{R2}
 R_3-R_2 \le \frac{R_2}{6(l-1).}
\end{equation}
Let $p:=(R_3,0)$. Let $\ell_1$ be the straight line passing through $a_1'$ and $p$;
and let $\ell_2$ be the straight line passing through $b_1'$ and $p$. Choose
$R_1$ sufficiently large with respect to $R_2$, so that $\ell_1$ and $\ell_2$
intersect $D_2'$. Let $b_2'$ be the point of intersection of $\ell_1$ and $D_2'$;
and let $a_2'$ be the point of intersection of $\ell_2$ and $D_2'$.
Let $D_3$ be the circle centered a the origin with radius $R_3$.
We refer to the region in the angular region of $D_1'$ bounded by the annulus between $D_2$ and $D_3$ as the
\emph{collision zone}. Let $D_2''$ be a circular arc, contained in the arc in $D_2$,
from $a_2'$ to $b_2'$. See Figure \ref{fig:collision_zone}.
We have that 

\ruyquote{every pair of crossing line segments, each joining a point of $D_2''$
with a clause point, intersect in the collision zone. }
Choose $R_1$ sufficiently large with respect to $R_2$, and $D_2''$ sufficiently small so that

\ruyquote{a line passing through a point of $D_2'' $ and a clause vertex, intersects
the collision zone in a line segment of length at most $2(R_3-R_2)$. \label{q:long edges}}

Starting at $a_2'$ and ending at $b_2'$ place $p_1, \dots, p_{n}$,
equally spaced vertices along $D_2''$. For each $i=1,\dots,n$, centered at each $p_i$ place a variable 
gadget for $x_i$. We choose the radius $r_1$'s of each of these gadgets so that
for all $1 \le i \le n$, $1 \le j \le m_i$,
we have that

\ruyquote{
\[d(x_{ij},y_{ij})=d(x_{i(j+1)},w_{ij})=d(\overline{x_{ij}},y_{ij})=d(\overline{x_{ij}},w_{ij})= \delta,\]
\label{q:delta}}
for some constant $\delta >0$. Choose the $r_1$'s sufficiently small, so that 

\ruyquote{the distance between
two vertices in distinct variable gadgets is at least
 $2 \delta.$ \label{eq:varg_faraway}}
\noindent 
Simple geometric arguments show that 

\ruyquote{ the distance from a $Y$-vertex or a $W$-vertex to a literal vertex,
other than one of its two closest such points, is at greater than $2\delta$. \label{q:Y faraway}}

If necessary, rotate the variable gadgets so that the union of the sets of literal and clause vertices is
in general position.
Choose the $r_1$'s sufficiently small with respect to $R_1$ so that 

\ruyquote{every pair of crossing line segments, each joining a vertex
of a variable gadget with a clause vertex, intersect in the collision zone. \label{q:col_zone}}

Now, for every $1 \le j \le l$, in the directed line from the base to $q_j$, we place
a vertex $\beta_j$ at distance equal to $\delta$ from the base. These vertices are called \emph{base parking}
vertices.

Let $t$ be a positive integer (appropriate lower bounds on $t$ are given below). 
For every $j=1,\dots,l$, and every literal $z_1, z_2, z_3$ appearing
on $C_j$, we proceed as follows. 
Let $\ell$ be the straight line segment from $z_i$ to $q_j$. 
Starting at $z_i$ and ending at $q_j$, along $\ell$,  we place $t$  consecutive and equally spaced vertices $\pi_{ij}$; 
we call them \emph{breadcrumb} vertices. See locations of vertices $\beta_j$ and $\pi_{ij}$ in Figure \ref{fig:cycling_gadgets}.
Thus, the distance between two consecutive breadcrumb vertices is equal to $d(q_j,z_i)/(t-1)$. 


 \subsection*{Scenes schedule}
 
The time intervals of the scenes are defined as follows.
See Figure \ref{fig:Filming} for an illustration of the location and time of the scenes. Let $\varepsilon_1 < \varepsilon_2 < \delta$. 
We choose $R_1$ sufficiently large with respect to $R_2$ so that 
\begin{equation}
   \frac{3(R_3-R_2)}{R_1-R_2}\varepsilon_2 < \varepsilon_1. \label{eq:epsilon}
\end{equation}
\begin{itemize}
  \item[$a)$] At time $s_0:=R_2+r_1$ a scene of length $R_1$ starts at each of the $W$-vertices.
  These scenes end at time $t_0$. We call them \emph{W-scenes}.

  \item[$b)$]  At time $s_0$ a scene of length $R_1-(R_2+r_1+\delta)$ starts at each of the base parking vertices.
   We call them \emph{base scenes}.

    \item[$c)$]  At time $s_1:=t_0+\delta$ a scene of length $R_1$ starts at each of the literal
    vertices.  These scenes end at time $t_1$.  We call them  \emph{literal scenes}.
    
    \item[$d)$]  At time $s_2:=t_1+\delta$ a scene of length $R_1$ starts at each of the $Y$-vertices.
    These scenes end at time $t_3$.  We call them \emph{Y-scenes}.
    
  \end{itemize}
  
  Let $s_1':=t_1$; inductively for $i=2,\dots, l$,  let $s_i':=s_{i-1}'+3(R_3-R_2)$.
  
  \begin{itemize}
   \item[$e)$] For every $i=1,\dots,l-1$, starting at time $s_i'$ and ending at $s_{i+1}'$, there are $t$
   consecutive scenes, each of length $\varepsilon_1/t$, located at each of the \emph{indexed} literal vertices \emph{not} appearing
   on $C_i$. The time between the ending and the beginning of two of these consecutive scenes
   is the same for every such pair. Note that the total time of these scenes at a each of these vertices, between times
   $s_i'$ and $s_{i+1}'$, is equal to $\varepsilon_1$. These scenes are called \emph{literal parking scenes}. Let $t_2$ be the finishing
   time of the last of all these parking scenes. Thus, using (\ref{R2}) we have: 
   \begin{equation}
 t_2= t_1+3(l-1)(R_3-R_2) \le t_1+R_2/2.
   \label{eq:t2}
\end{equation}
   \item[$f)$] For  every $i=1,\dots,l$, and every \emph{indexed} literal vertex $z$ appearing
   in $C_i$, we place one scene of length $\varepsilon_2/t$ at each of the breadcrumb vertices in the line from 
   $z$ to $q_i$ as follows. Let $z=z_1, \dots, z_t=q_i$ be these vertices, in the order as they appear on the directed line from $z$ to $q_i$.
   The first scene is placed on vertex $z_1$ and starts at $s_i'$.
   Inductively for $j=2,\dots, t-1$ the $j$-th scene is placed at $z_j$
   and starts $d(z,q_i)/(t-1)$ time after the $(j-1)$-scene has ended.  Thus, the total time of these
   scenes for each such $z$ is equal to $\varepsilon_2$. We call these scenes \emph{travelling scenes}.

\item[$g)$] At time $s_3:=t_2+R_1+r_1+\varepsilon_2$ a scene of length $\frac{R_2}{2}-2r_1-2\delta-\varepsilon_2$ starts at each of the clause vertices. We call these scenes
   $\emph{clause scenes}$.
\end{itemize}

By (\ref{q:col_zone}), we can choose $t$  sufficiently large, and $R_1$ sufficiently large with respect to $t$, so that

\ruyquote{ during a time interval of length $\varepsilon_2/t$ a drone cannot partially
film two scenes from the set of travelling scenes union literal parking scenes. \label{q:one_scene}}
   

We are ready now to prove Lemma \ref{th:satisf}. 
We first show
the more easy direction of the proof
, which will also give the fundamental ideas
of how the reduction works. 

  \begin{figure}
 	\centering
 	\includegraphics[width=.8\textwidth]{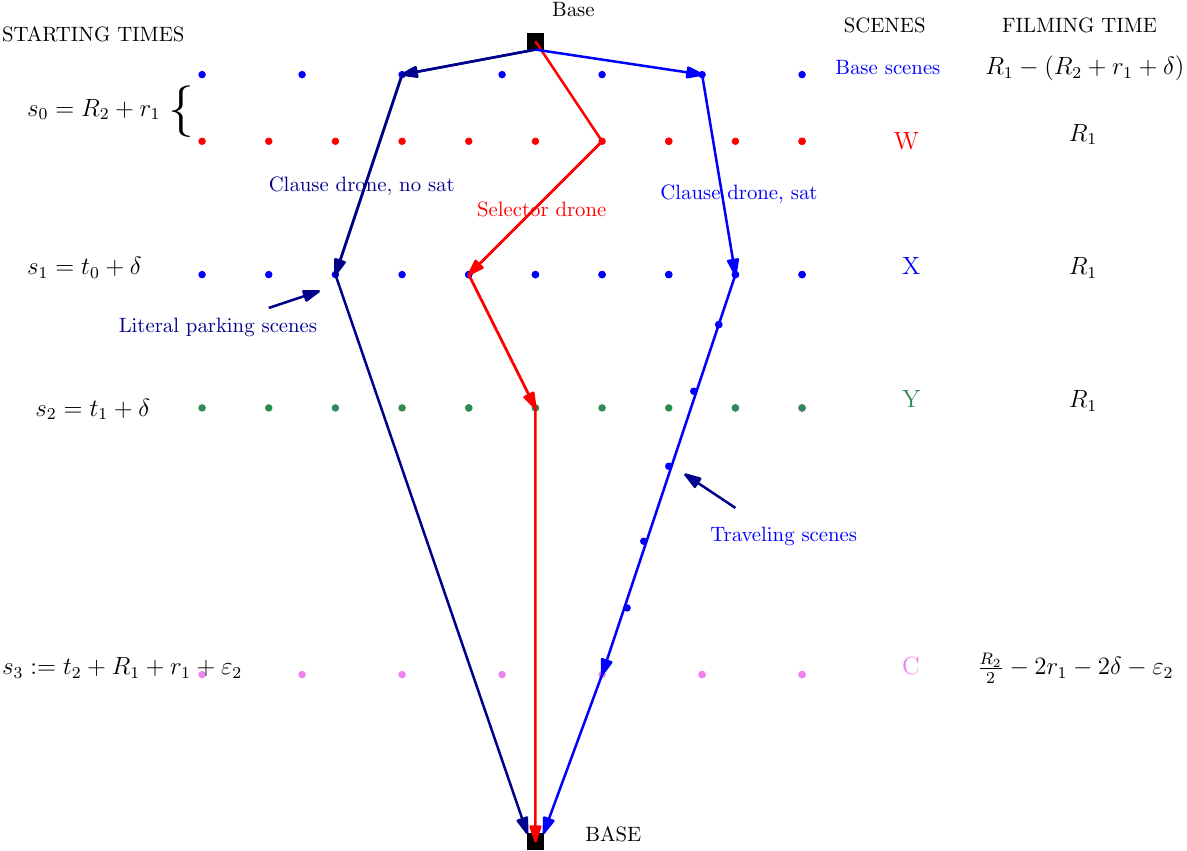}
 	\caption{Scenes and the flight plan.  }\label{fig:Filming}
 \end{figure}

\subsubsection*{$\Longrightarrow$ First, suppose that $\varphi$ is satisfiable.}
 Consider any assignment of the variables
$x_i$ that satisfies $\varphi$. We show  a flight plan of filming time equal to
\[T:=(3m+2l)R_1-\frac{l}{2}R_2-3lr_1-3l\delta+\frac{l(l-1)}{2}\varepsilon_1,\] using $k:=(m+l)$ drones with battery endurance equal to $L:=3R_1+2(R_2+r_1+\delta)$.

\subsubsection*{Flight plan}
  \begin{itemize}
      \item[$a)$] For each $w_{i,j}$ vertex, there is a drone arriving at $w_{i,j}$ at time $s_0$. These drones are called
     \emph{selector} drones; these $m$ drones film all the $W$-scenes.  The remaining drones are called \emph{clause} drones.
      This accounts for $m R_1$ filming time.

       \item[$b)$] For each $\beta_j$ vertex, there is a  clause drone arriving at $\beta_j$ at time $s_0$.
       These $l$ drones film all the base parking scenes.  At the end of their base parking scene, these drones return
       to the base.
      This accounts for $l(R_1-(R_2+r_1+\delta))$ filming time.

      \item[$c)$] For each variable $x_i$ we proceed as follows. If $x_i$ is set to ``True" then all the selector drones
      at the $x_i$-variable gadgets, after filming their $W$-scenes, move to their closest $\overline{x_{i,j}}$ vertex.
      If $x_i$ is set to ``False" then all the selector drones at the $x_i$-variable gadget, after filming their $W$-scenes, move to their closest $x_{i,j}$ vertex.
      These drones film the corresponding literal scene.
      This accounts for $m R_1$ filming time.

      \item[$d)$] After filming their corresponding literal scene, the selector drones move to their closest $Y$-vertex and film
      the corresponding $Y$ scene. This accounts for $m R_1$ filming time. After filming their corresponding $Y$-scene, the selector drones fly back to the base. 
      
    \end{itemize}
    For each clause $C_j$, pick a literal $z_j$ satisfying $C_j$. 
    \begin{itemize}
      \item[$e)$] At the indexed literal vertex corresponding to $z_j$,
      a clause drone arrives at time $s_1$ and films the literal scene. 
      Note that there is no selector drone at this vertex. This accounts for $l R_1$ filming time.

     \item[$f)$] For every $j=1,\dots,l$, the clause drone that arrived at $z_j$, stays at $z_j$ 
     from time $s_1'$ to time $s_j'$. This drones films a total
     of $(j-1)\varepsilon_1$ time, of the parking scenes at vertex $z_j$. In total summing over $j$, this amounts
     to a total of $\frac{l(l-1)}{2} \varepsilon_1$ of filming time.

    \item[$g)$]  For  every $j:=1,\dots,l$, at time $s_j'$ the clause drone at $z_j$ travels from $z_j$
   to $q_j$, in the process it films all the corresponding travelling scenes and at $q_j$ the drone
   films the clause scene. This accounts for a total of 
   $l\varepsilon_2+l\left (\frac{R_2}{2}-2r_1-2\delta-\varepsilon_2 \right)=l\left (\frac{R_2}{2}-2r_1-2\delta \right )$
   filming time. Finally, these drones fly back to the base.
  \end {itemize}

 The total filming time captured by the selector drones is equal to $3mR_1$
  and, the total filming time captured by the clause drones is equal to
  \[2lR_1-l(R_2+r_1+\delta)+\frac{l(l-1)}{2} \varepsilon_1+l\left (\frac{R_2}{2}-2r_1-2\delta \right).\]
  Thus the total filming times captured by the $k$ drones is equal to 
  \[(3m+2l)R_1-\frac{l}{2}R_2-3lr_1-3l\delta+\frac{l(l-1)}{2}\varepsilon_1=T.\]

We now elaborate on the battery. A selector drone, in the worst case, spends
\[\overset{\textrm{Go to W-vertex and film}}{\overbrace{R_2+r_1+R_1}}+\overset{\textrm{Go to literal and film}}{\overbrace{\delta+R_1}}+\overset{\textrm{Go to Y-vertex and film}}{\overbrace{\delta+R_1}}+\overset{\textrm{
Go home}}{\overbrace{R_2+r_1}}=3R_1+2(R_2+r_1+\delta)=L\]
of its battery endurance.
A clause drone spends \[\delta+R_1-(R_2+r_1+\delta)+\delta=R_1-R_2-r_1+\delta<L\]  of its battery endurance, when it first leaves the base,
films the base parking scene and returns to the base. Afterwards, in the worst case,
it spends 
\[\overset{\textrm{Go to literal}}{\overbrace{R_2+r_1}}+\overset{\textrm{Film literal}}{\overbrace{R_1}}+\overset{\textrm{
Parking scenes $\leq R_2/2$ by (\ref{R2})}}{\overbrace{3(l-1)(R_3-R_2)}}+\overset{\textrm{Go to clause}}{\overbrace{R_1+r_1}}+
\overset{\textrm{Travelling scenes}}
{\overbrace{\varepsilon_2}}+\overset{\textrm{Film clause}}{\overbrace{\frac{R_2}{2}-2r_1-2\delta-\varepsilon_2}}+
\overset{\textrm{
Go home}}{\overbrace{R_1}}\leq \]
\[3R_1+2R_2-2\delta+
<L\]
of its battery endurance.



\subsubsection*{$\Longleftarrow$ Second, let $F$ be a flight plan with filming time at least $T$ and battery endurance equal to $L$.}

We  first prove upper bounds on the total possible filming time per stages. Afterwards, we show that if a drone
misbehaves, then the flight plan fails to film this possible filming time by some amount, 
which we call a \emph{penalty}.  In the process, we refine the bounds on our parameters. 
The penalties are sufficiently large so that, combined with our upper bounds on the filming time,
ensure that $F$ behaves as above.

\vspace{.25cm}
 
We divide the scenes schedule in the following four stages. See the schedule in Figure \ref{fig:Filming}.

 \begin{itemize}
     \item \textbf{stage I}, starts at time $s_0$ and ends at time $t_0$;
     
        \item \textbf{stage II}, starts at time $s_1$ and ends at time $t_1$;
     
     \item \textbf{stage III}, starts at time $t_1$ and ends at time $t_2$;
     
     \item \textbf{stage IV}, starts at time $t_2$ and ends when the last drone returns to the base.
     \end{itemize}
 
 We have  that: 
 \begin{itemize} 
 \item $W$ and base parking scenes occur during stage I;  
 \item literal scenes occur during stage II;
 \item  literal parking scenes occur during stage III; 
  \item $Y$ and travelling scenes occur during stages III and IV; and
  \item clause scenes occur during stage IV.
 \end{itemize}
 
 \vspace{.25cm}

Our first upper bounds are immediate.
\begin{observation} \label{obs:stage_1}
The filming time of $F$ at stage I is at most 
\[(m+l) R_1 -l(R_2+r_1+\delta).\]
To achieve this bound a drone arrives at each of the $W$ and base parking vertices 
at time $s_0$; these drones do not move during Stage I.
\end{observation}

\begin{observation} \label{obs:stage_2}
The filming time of $F$ at stage II is at most 
\[(m+l)R_1.\]
To achieve this bound a drone arrives at each of the literal vertices
at time $s_1$; these drones do not move during Stage II.
\end{observation}

\begin{lemma} \label{lem:stage_3}
The filming time of $F$ at stage III is at most
\[m \left (t_2-t_1-\delta \right )+l(l-1)\varepsilon_1\]
\end{lemma}
\begin{proof}
At this stage, some amount of $Y-$scenes can be filmed by $m$ drones. 
The total time of $Y$-scenes during stage III is equal to
    \[m\left (t_2-t_1-\delta  \right ).\] We may assume that each of these scenes is filmed by a drone
    during stage III. 
    
    Consider one of the remaining $l$ drones. By (\ref{q:one_scene}), at every interval of length
    at most $\varepsilon_2 /t$, the drone can either film a travelling scene or a literal parking scene.
    Suppose that we are at time interval $ \left [s_i',s'_{i+1} \right ]$ for some $i=1,\dots,l-1$. This interval has
    duration equal to $3(R_3-R_2)$.
    Let $z$ be a literal vertex that satisfies $C_i$, and let $\ell$ be
    the line segment with endpoints $z$ and $q_j$.
    By (\ref{q:lower bound d}), the length
    of $\ell$ is lower bounded by $R_1-R_2$, and the total filming time of the travelling scenes
    playing along $\ell$ is equal to $\varepsilon_2$. Therefore, the total filming time of the travelling scenes
    playing at $\ell$ during time interval  $ \left [s_i',s'_{i+1} \right ]$ is at most
    \[\frac{3(R_3-R_2)}{R_1-R_2} \varepsilon_2.\]
    The total filming time in parking scenes that a drone can capture during time interval  $\left [s_i',s'_{i+1} \right ]$ 
    is at most $\varepsilon_1$.
    
    Let $\tau_1$ be the time that the drone is at one such $\ell$ during  time interval
    $[s_i',s'_{i+1}]$; and let $\tau_2$ be the time that the drone is at a literal vertex not satisfying
    $C_i$ during  time interval $\left [s_i',s'_{i+1}\right ]$. We have that $\tau_1+\tau_2 \le 3(R_3-R_2)$.
    By (\ref{eq:epsilon}), during time interval  $\left [s_i',s'_{i+1}\right ]$ the drone films at most
    \[\frac{3(R_3-R_2)}{R_1-R_2} \cdot \varepsilon_2 \cdot \frac{\tau_1}{3(R_3-R_2)}+\frac{\tau_2}{3(R_3-R_2)} \cdot \varepsilon_1  \le \varepsilon_1. \]
    Therefore, during this time interval the remaining $l$ drones film at most 
    $l\varepsilon_1$ and, the result follows.
\end{proof}
\begin{lemma}\label{lem:stage_4}
    The filming time of $F$ at stage IV is at most
    \[m(R_1 -t_2+t_1+\delta )+3l\varepsilon_2+l \left (\frac{R_2}{2}-2r_1-2\delta-\varepsilon_2 \right ).\]
\end{lemma}
\begin{proof}
    Each of $Y$-scene has a length $R_1 -t_2+t_1+\delta $ during stage IV. There are $m$ of these scenes. 
    The travelling scenes
    amount to $3 l\varepsilon_2$ time (at most 3 drones arrive to each clause). On the other hand, the clause scenes amount to $l \left (\frac{R_2}{2}-2r_1-2\delta-\varepsilon_2 \right )$ time. The result follows. 
\end{proof}

We point out that the upper bounds in Lemmas~\ref{lem:stage_3} and~\ref{lem:stage_4} are actually not achievable, 
but the bounds are enough for our purposes. We have the following general upper bound. 
Let $T'$ be the filming time of $F$.
\begin{corollary}\label{cor:upper bound}
$T'$ is at most
\[ (3m+2l)R_1-\frac{l}{2}R_2-3lr_1-3l\delta+2l\varepsilon_2+l(l-1)\varepsilon_1.\]
\end{corollary}
We have that \[T'-T \le 2l\varepsilon_2+\frac{l(l-1)}{2}\varepsilon_1.\]
Suppose that $R_1$ is sufficiently large with respect to $R_2$ and $l$ so that
\begin{equation} \label{q:R_1 diff}
    R_1-R_2-r_1>T'-T \le 2l\varepsilon_2+\frac{l(l-1)}{2}\varepsilon_1.
\end{equation}
By (\ref{q:R_1 diff}): every $W$-vertex and base parking vertex is visited at least
once by a drone during stage I. Moreover, we may assume that the first drone
to visit a given $W$-vertex arrived at time,
\[s_0+T'-T\le s_0+2l\varepsilon_2+\frac{l(l-1)}{2}\varepsilon_1,\] 
at the latest.
A drone that is the first to visit a given $W$-vertex during Stage I is called
a \emph{selector} drone;
a drone that is not a selector drone is  called \emph{clause} drone. There are exactly $m$ selector drones and $l$ clause drones.

If during Stage $I$, after arriving at its $W$-vertex, a selector drone flies back to the base, then this drone
films at most \[R_1-2(R_2-r_1)\]
of the $W$-scenes. If during Stage $I$, a clause drone flies to a $W$-vertex, then this drone
films at most \[R_1-(R_2+r_1+\delta)-2(R_2-r_1-\delta)\]
of the base parking scenes. We choose $R_2$ sufficiently large with respect to
$r_1$ and $l$ so that 
\begin{equation} \label{q:R_2 diff}
    2(R_2-r_1-\delta)>T'-T \le 2l\varepsilon_2+\frac{l(l-1)}{2}\varepsilon_1.
\end{equation}
By (\ref{q:R_2 diff}) we may assume that neither of these events occur. In particular

\ruyquote{at the beginning of Stage II, no selector drone has flown back to the base. \label{q:no going back II}}
\noindent By similar arguments, we have that 

\ruyquote{at the beginning of Stage III, no selector drone has flown back to the base. \label{q:no going back III}}

In particular, the selector drones have at most 
\[R_1-3 R_2+2r_1+2l \varepsilon_2+\frac{l(l-1)}{2}\varepsilon_1\]
battery endurance left at the beginning of Stage III.
Note that the travel time from a literal vertex to a clause vertex and then to the base
is at least \[2R_1-R_2-r_1.\]
This implies that

\ruyquote{a selector drone cannot film any part of a clause scene unless it flies to the base
after the beginning of Stage III.}

Suppose that at the beginning of Stage III (time $t_1$) a selector drone flies to the base
and then to a clause vertex. Its earliest time of arrival at the clause vertex
would be at least \[t_1+R_2-r_1+R_1.\] Notice that the ending time of the clause scenes
is equal to 
\[t_2+R_1+r_1+\varepsilon_2+\frac{R_2}{2}-2r_1-2\delta-\varepsilon_2 \le t_1+R_1+R_2-r_1-2\delta < t_1+R_2-r_1+R_1,\]
using (\ref{eq:t2}).
Therefore, 

\ruyquote{a selector drone cannot film any part of a clause scene.}
Also notice that a drone leaving from a $Y$-vertex at time $t_3$ towards a clause vertex, arrives
at time \[t_1+R_1+\delta+R_1-R_2=2R_1+\delta-R_2,\]
at the earliest.
Moreover, every $Y$-scene must be filmed for at least 
\[R_1-(T'-T) \ge R_1-2l\varepsilon_2-\frac{l(l-1)}{2}\varepsilon_1\] time.
Thus, we may assume that

\ruyquote{$Y$-scenes are only filmed by 
selector drones. \label{q:only Y}}
By (\ref{q:Y faraway}) and (\ref{q:only Y}), at time $t_0$, in every clause gadget corresponding 
to $x_i$, either all selector drones fly to their closest $x_{ij}$ vertex, or all
selector drones fly to their closest $\overline{x}_{ij}$ vertex. In the first case, we say that $x_i$ has been set to "False" and in the latter case,
we say that $x_i$ has been set to "True". Afterwards, they fly to their closest $Y$-vertex.
Thus, at time $s_1$, the clause drones arrive to literal vertices corresponding to their
truth value assignment. 

We elaborate now on the remaining drones. The goal is to show that different drones do not share travelling scenes.

Let $1 \le i <j \le l$. We choose $\varepsilon_2$ sufficiently small with respect to $R_3-R_2$ and $D_2''$ sufficiently small so that 

\ruyquote{ if a drone
films a part of a travelling scene corresponding to clause $C_j$, then it cannot
film a later travelling scene corresponding to clause $C_i$. \label{q:no past}}

Let $z$ be any literal vertex satisfying $C_j$.
We now choose $R_1$ sufficiently large, $R_3$ sufficiently close to $R_2$ and, $D_2''$ sufficiently small so that there exists
a distance $d^\ast$ that satisfies the following for all such $i,j$ and $z$. 


\ruyquote{If a drone films a travelling scene contained in the line passing through $z$ and $q_i$, and
at distance at least $d^\ast$ from $z$, then it cannot film any later travelling scene corresponding to
$C_j$. \label{q:max zigzag}}
Suppose that a drone leaves the literal vertex $z$ at time $s_i'$, and it eventually films
the clause scene at $q_j$. In the process, this drone films some travelling scenes. By (\ref{q:no past}) there exists 
$i=i_1 < i_2 < \dots,i_{l'} =j$, such that the drone films sequentially some travelling scenes corresponding to clauses 
$C_{i_1},\dots,C_{i_l'}$. Let $d_{s}$ be the minimum of the distances from the literal vertices satisfying $C_{i_s}$ to $q_{i_s}$. 
By (\ref{q:max zigzag}), the drone films at most 
\[\left (\sum_{s=1}^{l'-1} \frac{d^\ast}{d_{s}} \varepsilon_2 \right ) +\varepsilon_2 \le \frac{l d^\ast}{R_1-R_2}\varepsilon_2+\varepsilon_2. \]
We choose $R_1$ sufficiently large with respect to $R_2$ so that 
\begin{equation}
     \frac{l d^\ast}{R_1-R_2}\varepsilon_2 < \frac{\varepsilon_1}{l}. \label{q:no cheating}
\end{equation}

By (\ref{q:only Y}) and (\ref{q:no cheating}), we may update our upper bound on the filming of $F$ so that
\begin{equation} T'-T \le \frac{l(l-1)}{2}\varepsilon_1+\varepsilon_1. \label{q:diff2} \end{equation}
We choose $\varepsilon_2$ sufficiently large with respect to $\varepsilon_1$
so that \begin{equation} \varepsilon_2 > l^2 \varepsilon_1. \label{q:e1 vs e2} \end{equation}
We may assume that every clause scene is filmed by one clause drone.  
By (\ref{q:no past}), (\ref{q:no cheating}),  (\ref{q:diff2}) and (\ref{q:e1 vs e2}), we have that 

\ruyquote{the drone that films the clause scene at $q_i$ left its literal vertex at time $s_i'$ or earlier.}
\noindent 
Thus, the total amount of filming time due to literal parking scenes filmed by the  drone
that films the clause scene at $q_i$ is at most 
\[(i-1)\varepsilon_1.\]
In particular, the total amount of filming time due to literal parking scenes is upper bounded by
\[\sum_{i=1}^{l} (i-1) \varepsilon_1=\frac{l(l-1)}{2} \varepsilon_1.\]
Thus, we may further update our upper bound on $T'$ so that
\[T'-T \le \varepsilon_1.\]
This implies that the drone filming the clause scene at $q_i$ leaves its literal vertex
at time $s_i'.$ Thus, we have that 
$T=T'.$ Moreover, if this literal vertex does not satisfy $C_i$, then $T'< T$, since
it films less than $\varepsilon_2$ of the travelling scene corresponding to $q_i$. This completes the proof. 


	\section*{Acknowledgments}

This work is partially supported by grants PID2020-114154RB-I00 and TED2021-129182B-I00, funded by MCIN/AEI/10.13039/501100011033 and the European Union NextGenerationEU/PRTR.


\bibliographystyle{abbrv}
\bibliography{cycling}







\end{document}